%% file: main.tex
\documentclass[a4paper,UKenglish]{lipics/lipics}

\usepackage{microtype}
\usepackage{graphicx}
\input{packages}

\input{macros}

\title{Proving Looping and Non-Looping Non-Termination by Finite Automata\footnote{%
    This is an extended version of the paper~\cite{endr:zant:2015} published at RTA 2015.
    This extension includes 
    transformation for strengthening the presented non-termination techniques (see Remarks~\ref{rem:constant} and~\ref{rem:collapsing}), 
    a detailed description of the SAT encoding of the improved technique in Section~\ref{secimp} (see Remark~\ref{sat:improved})
    and a remark on the completeness of the method with respect to loops (see Remark~\ref{rem:loops}).
  }
}
                             
\author[1]{J\"{o}rg~Endrullis}
\author[2,3]{Hans~Zantema}

\affil[1]{
  Department of Computer Science, VU University Amsterdam,\\ 
  1081 HV Amsterdam, The Netherlands,
  email: {\tt j.endrullis@vu.nl}}
\affil[2]{
  Department of Computer Science, TU Eindhoven, P.O.\ Box 513,\\
  5600 MB Eindhoven, The Netherlands,
  email: {\tt H.Zantema@tue.nl}}
\affil[3]{
  Institute for Computing and Information Sciences, Radboud University \\
  Nijmegen, P.O.\ Box 9010, 6500 GL Nijmegen, The Netherlands
}  

\authorrunning{J. Endrullis and H. Zantema}

\Copyright{J. Endrullis and H. Zantema}


\serieslogo{}
\volumeinfo
 {Carsten Fuhs}
 {1}
 {14th International Workshop on Termination 2014 (WST~'14)}
 {1}
 {1}
 {1}
\EventShortName{}

\def\copyrightline{}

\begin{document}

\maketitle

\begin{abstract}
A new technique is presented to prove non-termination of term rewriting. The basic idea is to find a non-empty 
regular language of terms that is closed under rewriting and does not contain normal forms. It is automated by
representing the language by a tree automaton with a fixed number of states, and expressing the mentioned requirements
in a SAT formula. Satisfiability of this formula implies non-termination. 
Our approach succeeds for many examples where all earlier techniques fail, for instance for the $S$-rule from combinatory
logic.
\end{abstract}

\section{Introduction}
\input{intro}

\section{Abstract Rewriting}
\label{secar}
\input{abstract_rewriting}

\section{Tree Automata}
\label{secta}
\input{automata}

\section{Basic Methods}
\label{secbm}

\input{methods}

\section{Improved Methods for Disproving Strong Normalization}
\label{secimp}
\input{improved}

\section{Experimental Results}
\label{secres}
\input{results}

\section{Conclusions and Future Work}
\label{secconc}
\input{conclusions}

\bibliography{main.bib}

\end{document}

%% file: packages.tex
\usepackage{amsmath,amssymb,amstext}
\usepackage{proof}
\usepackage{wrapfig}
\usepackage{framed}
\usepackage{enumitem}

%% file: macros.tex
\theoremstyle{definition}
\newtheorem{technique}[theorem]{Technique}
\newtheorem{encoding}[theorem]{Remark}

\newcommand{\tuple}[1]{\langle\,#1\,\rangle}

\newcommand{\ter}[2]{T(#1,#2)}
\newcommand{\vars}{\mathcal{X}}
\newcommand{\varsof}[1]{\mathit{Var}(#1)}
\newcommand{\pos}[1]{\mathit{Pos}(#1)}

\newcommand{\pow}[1]{\mathcal{P}(#1)}

\newcommand{\lang}[1]{\mathcal{L}(#1)}

\renewcommand{\emptyset}{\varnothing}
\newcommand{\nat}{\mathbb{N}}

\newcommand{\arity}[1]{\#(#1)}

\newcommand{\sap}{a}
\newcommand{\ap}[2]{\sap(#1,#2)}
\newcommand{\combS}{\mathit{S}}

\newcommand{\rmap}{\xi}
\newcommand{\ared}{\rightsquigarrow}

\newcommand{\nb}{\nobreakdash}

\newcommand{\powerset}[1]{\wp(#1)}

\newcommand{\product}{\times}

\newcommand{\SN}{\mbox{\sf SN}}
\newcommand{\WN}{\mbox{\sf WN}}
\newcommand{\NF}{\mbox{\sf NF}}

%% file: intro.tex
A basic approach for proving that a term rewriting system (TRS) is non-terminating is to prove that it admits a {\em
loop}, that is, a reduction of the shape $t \to^+ C[t \sigma]$, see
\cite{gese:zant:1999}. Indeed, such a loop gives rise to an infinite
reduction $t \to^+ C[t \sigma] \to^+ C[(C[t \sigma]) \sigma] \to \cdots$ in which in every step $t$ is replaced by $C[t
\sigma]$. In trying to prove non-termination, several tools (\cite{aprove,ttt2}) search for a loop. An extension from
\cite{emme:enge:gies:2012}, implemented in \cite{aprove} goes a step further: it searches for
reductions of the shape  
$t \sigma^n \mu \to^+ C[t \sigma^{f(n)} \mu \tau]$ for every $n$ for a linear increasing function $f$, and some extensions.
All of these patterns are chosen to be extended to an infinite reduction in an 
obvious way, hence proving non-termination. However, many non-terminating TRSs exist not admitting an infinite reduction 
of this regular shape, or the technique from  \cite{emme:enge:gies:2012} fails to find it. 

A crucial example is the $S$-rule $a(a(a(S, x), y), z) \to a(a(x, z), a(y, z))$, one of the building blocks of
Combinatory Logic. Although being only one single rule, and having nice properties like orthogonality, non-termination
of this system is a hard issue. Infinite reductions are known, but are of a complicated shape, see \cite{wald:2000}. So developing a
general technique that can prove non-termination of the $S$-rule automatically is a great challenge. In this paper we succeed in
presenting such a technique, and we describe a SAT-based approach by which non-termination of many TRSs, including 
the $S$-rule, is proved fully automatically.

The underlying idea is quite simple: non-termination immediately follows from the existence of a non-empty set of 
terms that is closed under rewriting and does not contain normal forms. Our approach is to find such a set being the
language accepted by a finite tree automaton, and find this tree automaton from the satisfying assignment of a SAT
formula describing the above requirements. Hence the goal is to describe the requirements, namely non-emptiness, closed
under rewriting, and not containing normal forms, in a SAT formula. 

We want to stress that having quick methods for proving non-termination of term rewriting also may be fruitful for proving
termination. In a typical search for a termination proof, like using the dependency pair framework, the original  
problem is transformed in several ways to other termination problems that not all need to be terminating. Being able to
quickly recognize non-termination of some of them makes a further search for termination proofs redundant, which may speed
up the overall search for a termination proof.

We note that, like termination, non-termination is an undecidable property.
However, while termination is $\Pi^0_2$-complete,
non-termination is $\Sigma^0_2$-complete~\cite{endr:geuv:zant:2009,endr:geuv:simo:zant:2011}.

The paper is organized as follows. In Section \ref{secar} we present our basic approach in
the setting of abstract reduction systems on a set $T$, in which the language is just a
subset of $T$. Surprisingly, being not weakly normalizing corresponds to (strongly) closed
under rewriting, and being not strongly normalizing corresponds to weakly closed under
rewriting. In Section \ref{secta} we give preliminaries on tree automata and show how string
automata can be seen as an instance of tree automata. In Section \ref{secbm} we present our
basic methods, starting by how the requirements are expressed in SAT, and next how this is
used to disprove weak normalization and strong normalization. In Section \ref{secimp} we
strengthen our approach by labeling the states of the tree automata by sets of rewrite rules 
and exploiting this in the method. In Section \ref{secres} we present experimental results of
our implementation. We conclude in Section \ref{secconc}.

\paragraph*{Related Work}

The paper~\cite{gese:zant:1999} introduces the notion of loops and investigates necessary conditions for the existence of
them.  The work~\cite{zank:midd:2007} employs SAT solvers to find loops, 
\cite{zank:ster:hofb:midd:2010} uses forward closures to find loops efficiently, and
\cite{wald:2012} introduces `compressed loops' to find certain forms of very long loops. 
Non-termination beyond loops has been investigated in~\cite{oppe:2008} and~\cite{emme:enge:gies:2012}. There
the basic idea is the search for a particular generalization of loops, like a term $t$ and substitutions $\sigma, \tau$
such that for every $n$ there exist $C, \mu$ such that $t \sigma^n \tau$ rewrites to $C[t \sigma^{f(n)} \tau \mu]$, for some
ascending linear function $f$. Although the $S$-rule admits such reductions, these techniques fail to find them.
For other examples for which not even reductions exist of the shape studied in ~\cite{oppe:2008}
and~\cite{emme:enge:gies:2012}, we will be able to prove non-termination fully automatically.

Our approach can be summarized as searching for non-termination proofs based on regular (tree) automata.
Regular (tree) automata have been fruitfully applied to a wide rage of properties of term rewriting systems: 
for proving termination~\cite{gese:hofb:wald:zant:2007,endr:hofb:wald:2006,korp:midd:2009},
infinitary normalization~\cite{endr:grab:hend:klop:vrij:2009},
liveness~\cite{mous:lada:zant:2010},
and for analyzing reachability and deciding the existence of common reducts~\cite{felg:thie:2014,endr:grab:klop:oost:2011}.
Local termination on regular languages, has been investigated in~\cite{endr:vrij:wald:2010}.

%% file: abstract_rewriting.tex
An {\em abstract reduction system} (ARS) is a binary relation $\to$ on a set $T$.
We write 
$\to^+$ for the transitive closure, and
$\to^*$ for the reflexive, transitive closure of $\to$.

Let $\to$ be an ARS on $T$.
The ARS $\to$ is called {\em terminating} or {\em strongly normalizing} ($\SN$) if no infinite sequence
$t_0,t_1,t_2,\ldots \in T$ exists such that $t_i \to t_{i+1}$ for all $i \geq 0$.
A {\em normal form} with respect to $\to$ is an element $t \in T$ such that no $u \in T$ exists satisfying $t \to u$.  
The set of all normal forms with respect to $\to$ is denoted by $\NF(\to)$. 
The ARS $\to$ is called {\em weakly normalizing} ($\WN$) if for every $t \in T$ a normal form $u \in T$ 
exists such that $t \to^* u$.

\begin{definition}
  A set $L \subseteq T$ is called
  \begin{itemize}
    \item {\em closed under} $\to$ if for all $t \in L$ and all $u \in T$ satisfying $t \to u$ it holds $u \in L$, and
    \item {\em weakly closed under} $\to$ if for all $t \in L \setminus \NF(\to)$ there exists $u \in L$ such that $t \to u$.
  \end{itemize}
\end{definition}

It is straightforward from these definitions that 
if $L$ is closed under $\to$, then $L$ is weakly closed under $\to$ as well.
The following theorems relate these notions to $\SN$ and $\WN$.

\begin{theorem}\label{thm:sn}
  An ARS $\to$ on $T$ is not $\SN$ if and only if a non-empty $L \subseteq T$ exists such that 
  $L \cap \NF(\to) = \emptyset$ and $L$ is weakly closed under $\to^+$.
\end{theorem}
\begin{proof}
  If $\to$ is not $\SN$ then an infinite sequence
  $t_0,t_1,t_2,\ldots \in T$ exists such that $t_i \to t_{i+1}$ for all $i \geq 0$. Then $L = \{t_i \mid i \geq 0 \}$
  satisfies the required properties.
  
  Conversely, assume $L$ satisfies the given properties. Since $L$ is non-empty we can choose $t_0 \in L$, and using the
  other properties for $i = 0,1,2,3,\ldots$ we can choose $t_{i+1} \in L$ such that $t_i \to^+ t_{i+1}$, proving that 
  $\to$ is not $\SN$.
\end{proof}  

\begin{theorem}\label{thm:wn}
  An ARS $\to$ on $T$ is not $\WN$ if and only if a non-empty $L \subseteq T$ exists such that 
  $L \cap \NF(\to) = \emptyset$ and $L$ is closed under $\to$.
\end{theorem}
\begin{proof}
  If $\to$ is not $\WN$ then $t \in T$ exists such that $L \cap \NF(\to) = \emptyset$ for $L = \{u \in T \mid t \to^* u\}$. 
  Then $L$ satisfies the required properties.
  
  Conversely, assume $L$ satisfies the given properties. Since $L$ is non-empty we can choose $t_0 \in L$. 
  Assume that $\to$ is $\WN$, then  $t_0 \to t_1 \to \cdots \to t_n$ exists such that $t_n \in \NF(\to)$. Since
  $L$ is closed under $\to$ we obtain $t_i \in L$ for $i = 1,2,\ldots,n$, contradicting $L \cap \NF(\to) = \emptyset$.
\end{proof}

A variant of Theorem~\ref{thm:sn}, where $\to^+$ is replaced by $\to$, has been observed in~\cite{cook}.
To the best knowledge of the authors, Theorem~\ref{thm:wn} has not been observed in the literature.

%% file: automata.tex
\begin{definition}\label{def:nfa}
  A \emph{(non-deterministic finite) tree automaton $A$ over a signature $\Sigma$} 
  is a tuple $A = \tuple{Q,\Sigma,F,\delta}$
  where 
  \begin{enumerate}[label=(\emph{\roman*})]
    \item $Q$ is a finite set of \emph{states},
    \item $F \subseteq Q$ is a set of \emph{accepting states}, and
    \item $\delta$ a set of rewrite rules, called \emph{transition rules}, of the shape
      \begin{align*}
        f(q_1,\ldots,q_n) \ared q
      \end{align*}
      where $n$ is the arity of $f \in \Sigma$ and $q_1,\ldots,q_n,q \in Q$. 
      We write $\ared$ for the rewrite relation generated by the rules $\delta$.
  \end{enumerate}
\end{definition}
Note that we use $\ared$ to distinguish automata transitions
from term rewriting $\to$ with respect to some TRS $R$.

\begin{definition}
  The \emph{language $\lang{A}$ accepted by $A$} is the set
  \begin{align*}
    \lang{A} = \{ \,t \mid t \in \ter{\Sigma}{\emptyset},\; q \in F,\; t \ared^* q \,\}
  \end{align*}
  of ground terms that rewrite to a final state.
\end{definition}

The kind of tree automata considered here is called {\em bottom up} in the literature.
Sometimes in the definition of bottom-up tree automaton the right hand side $q$ in the rule
has arguments and the acceptance criterion is rewriting to a term with a final state as root.
However, when tree automata are only used for defining (term) languages as is the case in
this paper, these definitions coincide. 

Tree automata can be seen as a generalization of string automata as follows. For a 
string automaton ($=$ NFA) $S$  define the tree automaton $A$ by 
\begin{itemize}
\item taking the same sets of states and accepting states, and 
\item taking as signature the same signature in which all symbols are unary, extended by a
single constant $\varepsilon$, and
\item taking as transition rules $\varepsilon \ared q_0$ for $q_0$ being the initial state of
$S$, and for every transition $q \stackrel{a}{\to} q'$ in $S$ the rule
$a(q) \ared q'$.
\end{itemize}
Form this definition it is immediate that a string $a_1 a_2 \cdots a_n$ is accepted by $S$ if 
and only if $a_n(a_{n-1}(\cdots(a_1(\varepsilon))\cdots))$ is accepted by $A$. 
Here we assume that $S$ reads the string from left to right
(otherwise there is no need to reverse the order of the letters).


\begin{example}\label{ex:automaton:lr}
  To define a tree automaton accepting the language $b\;a^*\;(L|R)\;a^*\;b$, that is,
  all words that start with $b$, end with $b$, contain one $L$ or $R$ and otherwise only $a$,
  we start by its corresponding string automaton

\noindent
\begin{minipage}{78mm}
\includegraphics[scale=0.19]{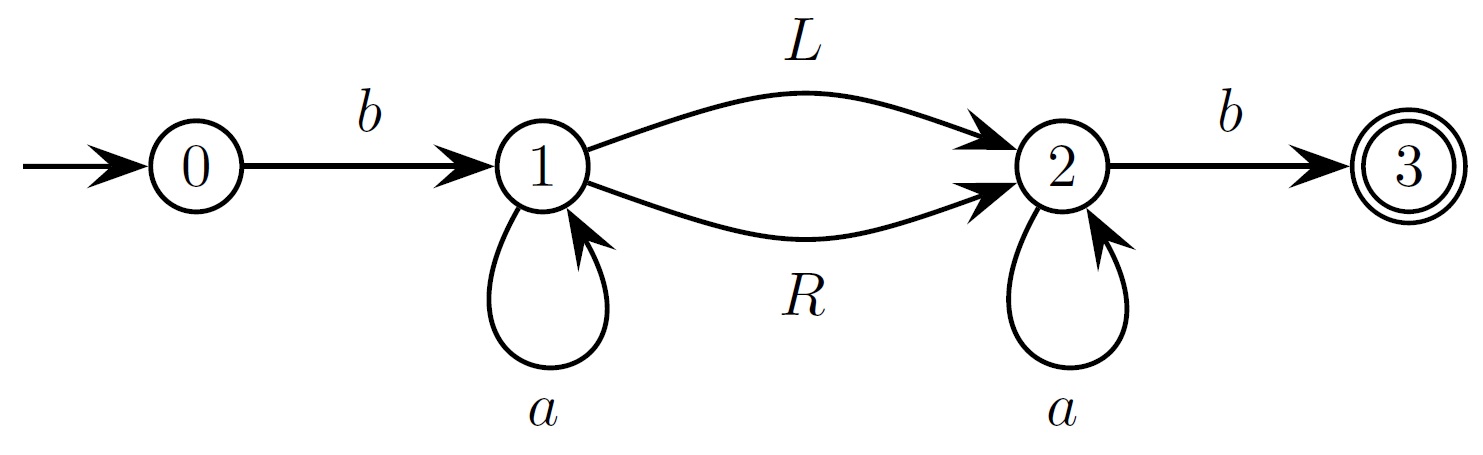}
\end{minipage}
\begin{minipage}{6cm}
The above construction yields the tree automaton $A_{LR} = \tuple{Q,\Sigma,F,\delta}$ where 
  $\Sigma = \{b,L,R,a,\varepsilon\}$
  in which $b,L,R,a$ are unary and $\varepsilon$ is a constant, $Q = \{0,1,2,3\}$, $F = \{3\}$
  and $\delta$ consists of the rules
\end{minipage}
  \begin{align*}
    \varepsilon &\ared 0 & a(1) &\ared 1 & b(0) &\ared 1 & R(1) &\ared 2 & L(1) &\ared 2\\
    && a(2) &\ared 2 & b(2) &\ared 3 
  \end{align*}
\end{example}

\begin{example}\label{ex:automaton:s}
  The following is a tree automaton for the signature $\Sigma = \{\sap,\combS\}$ 
  where $a$ is binary and $\combS$ is a constant.
  Let $A_{S} = \tuple{Q,\Sigma,F,\delta}$ where $Q = \{0,1,2,3,4\}$, $F = \{4\}$
  and
  \begin{align*}
    \combS &\ared 0 & 
    \sap(0,0) &\ared 1 &
    \sap(1,0) &\ared 2 & 
    \sap(2,2) &\ared 3 & 
    \sap(3,3) &\ared 3 
    \\ 
    && \sap(0,2) &\ared 2 &
    && \sap(2,3) &\ared 3 & 
    \sap(3,3) &\ared 4
    \\ 
    && \sap(0,3) &\ared 2
  \end{align*}
As is usual in combinatory logic, ground terms are represented by omitting the $a$ symbol
and writing $uvw = (uv)w$.  We show that this automaton accepts the term 
  $\combS\combS\combS(\combS\combS\combS)(\combS\combS\combS(\combS\combS\combS))$:
  \begin{align*}
    \combS\combS\combS(\combS\combS\combS)(\combS\combS\combS(\combS\combS\combS)) 
    &\ared^{12} 000(000)(000(000)) \\ 
    &\ared^4 10(10)(10(10)) 
    \ared^4 22(22) 
    \ared^2 33 
    \ared^1 4 
  \end{align*}
  Since $4 \in F$ the term is accepted by the automaton.

  This automaton has been found automatically by our tool, 
  and its language is closely related to the $\mathcal{QQQ}$-criterion of Waldmann~\cite{wald:2000,starling:2015}.
  Roughly speaking, the language recognized by this automaton can be described as follows:
  \begin{itemize}
    \item state $0$ accepts only the term $\combS$,
    \item state $1$ accepts only the term $\combS\combS$,
    \item state $2$ corresponds to terms that contain at least \emph{one} occurrence of $\combS\combS\combS$,
    \item state $3$ corresponds to terms that contain at least \emph{two} occurrence of $\combS\combS\combS$, and
    \item state $4$ accepts terms $MN$ for which both $M$ and $N$ contain two occurrences of $\combS\combS\combS$.
  \end{itemize}
\end{example}

%% file: methods.tex
In this section, we are concerned with
automating the abstract non-termination methods from Section~\ref{secar}.
To this end, we use finite tree automata giving rise to regular tree languages.
We first develop methods for disproving weak normalization
and then for disproving strong normalization.

The applicability of the non-termination techniques described in the remainder of this paper 
can be improved by two simple transformations of term rewrite systems $R$.
The first transformation concerns the introduction of a fresh constant in $\Sigma$ (see Remark~\ref{rem:constant}) 
and the second transformation describes the elimination of collapsing rules in $R$ (Remark~\ref{rem:collapsing}).
Both transformation do not affect (weak) normalisation of $R$.

\begin{encoding}[Adding fresh constants]\label{rem:constant}
  Let $R$ be a TRS over a signature $\Sigma$.
  If we are interested in non-termination of $R$ on arbitrary terms (including non-ground terms),
  then, without loss of generality, we can extend the signature~$\Sigma$
  with a fresh constant $c$ (a symbol of arity~$0$).
  Since the fresh constant can be thought of as a variable,
  this transformation neither affects weak nor strong normalisation of $R$.
  Moreover, by adding a constant, we can reduce strong normalisation of $R$ 
  to strong normalisation of $R$ on all ground terms.
  The reason is that, if the signature contains at least one constant symbol, then both properties coincide.
  
  Note that adding a single fresh constant does not not suffice for weak normalisation.
  The addition of multiple fresh constants might be needed to make weak normalisation of $R$ 
  coincide with weak normalisation of $R$ on ground terms.
  
  In this paper, we are interested in disproving (weakly) normalisation.
  The automata techniques we employ in this paper actually yield the the stronger property
  that $R$ is not (weakly) normalising on ground terms.
  Therefore we tacitly assume that the signature is extended with a fresh constant (if it does not already contain one).
\end{encoding}

\begin{encoding}[Elimination of collapsing rules]\label{rem:collapsing}
  Let $R$ be a TRS over a signature $\Sigma$.
  For the techniques in the remainder 
   by first eliminating collapsing rules,
  that is, rules of the form $\ell \to x$ with $x \in \vars$.
  Assume that the TRS $R$ contains a collapsing rule $\ell \to x$.
  For every $f \in \Sigma$ we define the substitution $\sigma_f : \vars \to \ter{\Sigma}{\vars}$ 
  by $\sigma_f(x) = f(x_1,\ldots,x_{\arity{f}})$ for fresh variables $x_1,\ldots,x_{\arity{f}}$
  and $\sigma_f(y) = y$ for all $y \ne x$.
  We define 
  \begin{align*}
    R' = (R \setminus \{\ell \to x\}) \cup \{\ell\sigma_f \to x\sigma_f \mid f \in \Sigma\}
  \end{align*}
  Then $R$ and $R'$ induce the same rewrite relation on ground terms.
  Hence $R'$ is (weakly) ground normalizing if and only if $R$ is.

  Moreover, $R'$ is strongly normalising if and only if $R$ is.
  This can be seen as follows.
  The `if'-direction follows immediately from ${\to_{R'}} \subseteq {\to_R}$.
  For the `only if'-direction assume that $R$ admits an infinite rewrite sequence $t_1 \to_R t_2 \to_R t_3 \to_R \cdots$.
  Let $f \in \Sigma$ and let $x_1,\ldots,x_{\arity{f}}$ be fresh variables.
  Define a substitution $\sigma$ by $\sigma(x) = f(x_1,\ldots,x_{\arity{f}})$ for all $x \in \vars$.
  Then $t_1\sigma \to_{R'} t_2\sigma \to_{R'} t_3\sigma \to_{R'} \cdots$ is an infinite rewrite sequence in $R'$.
\end{encoding}

\subsection{SAT Encoding of Properties}
\input{automata-sat.tex}

\subsection{Disproving Weak Normalization}

We are now ready to use Theorem~\ref{thm:wn} in combination with tree automata
for automatically disproving weak normalization.
The language $L$ in the theorem is described by a non-deterministic tree automaton.
In the previous section, we have seen how the relevant properties of
tree automata can be checked.
Here, we summarize the procedure:

\begin{technique}\label{tec:aut:wn}
  Let $R$ be a left-linear TRS. We search for a tree automaton $A = \tuple{Q,\Sigma,F,\delta}$
  such that $\lang{A}$ fulfills the properties of Theorem~\ref{thm:wn}:
  \begin{enumerate}[label=(\emph{\roman*})]
    \item We guarantee $\lang{A} \cap \NF(\to) = \emptyset$
      by the following steps:
      \begin{itemize}
        \item We employ Lemma~\ref{lem:redex:automaton} to construct a deterministic, complete automaton $B = \tuple{Q,\Sigma,F,\delta}$ 
          that accepts the set of terms containing redex occurrences with respect to $R$.
        \item Then the automaton $\overline{B} = \tuple{Q,Q \setminus \Sigma,F,\delta}$ accepts all ground normal forms.
        \item We use Lemma~\ref{lem:intersection} to check that $\lang{A} \cap \lang{\overline{B}} = \varnothing$ (thus $\lang{A} \subseteq \lang{B}$).
      \end{itemize} 
    \item We guarantee that $\lang{A}$ is closed under $\to$ by Lemma~\ref{lem:closed}.
    \item We use Lemma~\ref{lem:empty} to ensure that $\lang{A} \ne \emptyset$.
  \end{enumerate}
  These conditions can be encoded as satisfiability problems 
  as described in Remarks~\ref{sat:automaton}, \ref{sat:reachable}, \ref{sat:intersection}, \ref{sat:empty} and \ref{sat:closed}.
  This enables us to utilize SAT solvers to search for suitable automata $A$.
\end{technique}

\begin{encoding}
  We note the combination of Technique~\ref{tec:aut:wn} with Remark~\ref{rem:collapsing} is complete with respect to 
  disproving weak normalization on regular languages:
  if there exists a regular language $L$ fulfilling the conditions of Theorem~\ref{thm:wn},
  then weak normalization can be disproved using Technique~\ref{tec:aut:wn} after eliminating collapsing rules as in Remark~\ref{rem:collapsing}.

  This can be seen as follows.
  In the work~\cite{endr:vrij:wald:2009,endr:vrij:wald:2010} a generalized method
  for ensuring closure of the language of automata under rewriting has been proposed.
  Thereby the condition $\ell\alpha \ared_A^* q \;\implies\; r\alpha \ared_A^* q$ of Lemma~\ref{lem:closed} 
  is weakened to 
  \begin{align}
    \ell\alpha \ared_A^* q \;\implies\; r\alpha \ared_A^* p \quad \text{for some $p \ge q$}\;. \label{eq:quasi:order}
  \end{align}
  Here $\le$ is a quasi-order on the states $Q$
  and the automaton must be monotonic with respect to this order, see Definition~\ref{def:monotonic}.
  The monotonicity guarantees that the language of the automaton is closed under rewriting.

  In~\cite{felg:thie:2014} it has been shown that this monotonicity property 
  is strong enough to characterize and decide the closure of the regular languages under rewriting.
  In particular, the language of a deterministic tree automaton is closed under rewriting 
  if and only if there exists such a monotonic quasi-order on the states.  

  Let $R$ be a TRS such that there exists a regular language that satisfies the conditions of Theorem~\ref{thm:wn}.
  Then there exists a deterministic, complete automaton $A$ accepting this language
  and a quasi-order $\le$ on the states satisfying \eqref{eq:quasi:order} and monotonicity.
  Let $R'$ be obtained from $R$ by eliminating collapsing rules as described in Remark~\ref{rem:collapsing}.
  We obtain a non-deterministic automaton $A'$ 
  that fulfils the requirements of Technique~\ref{tec:aut:wn} for $R'$
  by closing the transition relation of $A$ under $\le$:  
  we add $f(q_1,\ldots,q_n) \ared q$ whenever $q \le p$ and $f(q_1,\ldots,q_n) \ared p$.
  As a consequence of monotonicity and using induction over the term structure,
  we obtain for all terms $t \in \ter{\Sigma}{\vars}$ with $t \not\in \vars$
  and $\alpha : \vars \to Q$ that 
  \begin{itemize}
    \item [($\star$)] $t \ared_{A'}^* q$ if and only if $t \ared_{A}^* p$ for some $p$ with $q \le p$.
  \end{itemize}
  As a consequence of ($\star$) and monotonicity we have $\lang{A'} = \lang{A}$
  (roughly speaking, if $q \le p$, then $q$ accepts a subsets of the language of $p$).
  Thus $\lang{A'} \cap \NF(\to) = \emptyset$ and $\lang{A'} \ne \emptyset$ are guaranteed.
  Finally, we show that Lemma~\ref{lem:closed} is applicable for $R'$ and $A'$.
  Let $\ell \to r \in R'$, $\alpha : \vars \to Q$ and $q \in Q$ such that $\ell\alpha \ared_{A'}^* q$.
  Then by ($\star$) we get $\ell\alpha \ared_{A}^* q'$ for some $q' \in Q$ with $q \le q'$.
  By \eqref{eq:quasi:order} we have that $r\alpha \ared_{A}^* q''$ for some $q'' \in Q$ with $q' \le q''$.
  Again by ($\star$) we obtain that $r\alpha \ared_{A}^* q$.
  Hence the conditions of Technique~\ref{tec:aut:wn} are fulfilled for $R'$ and $A'$.
\end{encoding}

\begin{example}\label{ex:lr}
  We consider the following string rewriting system:
  \begin{align*}
   aL &\to La & 
   Ra &\to aR &
   bL &\to bR &
   Rb &\to Lab
  \end{align*}
  This rewrite system is neither strongly nor weakly normalizing,
  but does not admit looping reductions, that is, reductions of the form $s \to^+ \ell s r$.
  An example of an infinite reduction is:
  \begin{align*}
    bLb \to bRb \to bLab \to bRab \to baRb \to baLab \to bLaab \to bRaab \to \cdots 
  \end{align*}
  It is easy to check that the automaton $A_{LR}$ from Example~\ref{ex:automaton:lr}
  fulfills the requirements of Technique~\ref{tec:aut:wn}.
  Hence, the system is not weakly normalizing.
\end{example}

\begin{example}\label{ex:S}
  We consider the $S$-rule from combinatory logic:
  \begin{align*}
    \ap{\ap{\ap{\combS}{x}}{y}}{z} \to \ap{\ap{x}{z}}{\ap{y}{z}}
  \end{align*}  
  For the $S$-rule it is known that there are no reductions $t \to^* C[t]$ for ground terms $t$, see~\cite{wald:2000}.
  For open terms $t$ the existence of reductions $t \to^* C[t\sigma]$ is open.
  
  It is straightforward to verify that the automaton $A_S$ from Example~\ref{ex:automaton:s}
  fulfills the requirements of Technique~\ref{tec:aut:wn},
  and hence the $S$-rule,
  and in particular the term
  $\combS\combS\combS(\combS\combS\combS)(\combS\combS\combS(\combS\combS\combS))$,
  are not weakly normalizing.
  
\end{example}

\begin{example}\label{ex:delta}
  The $\delta$-rule (known as Owl in Combinatory Logic) is even simpler:
  \begin{align*}
    \delta x y \to y(xy), && \text{ or equivalently } && \ap{\ap{\delta}{x}}{y} \to \ap{y}{\ap{x}{y}}\;.
  \end{align*}
  As shown in~\cite{starling:owl:2015}, this rule does not admit loops, , and
  the techniques in~\cite{emme:enge:gies:2012} fail for this system. 
  The Technique~\ref{tec:aut:wn} can be applied to automatically disprove weak-normalization for this rule.
  Our tool finds a tree automaton that has $3$ states and
  accepts the language of all ground terms with two occurrences of $\delta\delta$.
  In fact, this is precisely the language of non-terminating ground $\delta$-terms, see further~\cite{starling:owl:2015}.
\end{example}

In all examples until now infinite reductions exist of the regular shape based on
$t \sigma^n \tau$ rewriting to a term having $t \sigma^{f(n)} \tau \mu$ as a sub-term, 
for every $n$, for some term $t$ and substitutions $\sigma, \tau, \mu$ and an
ascending linear function $f$. For instance, the $S$ rule (Example \ref{ex:S}) admits
an infinite reduction implied by $t \sigma^n \tau$ rewriting to a super-term of 
$t \sigma^{n+1} \tau$, for $t = a(x,x)$, $\sigma(x) = Ax$, $\tau(x) = SA(SAA)$, for 
$A = SSS$. 
\begin{example}\label{ex:lr:aaa}
The following example does not have an infinite reduction of this regular
shape, neither of the more general patterns from ~\cite{oppe:2008} and
\cite{emme:enge:gies:2012}. 
  \begin{align*}
   aL &\to La & 
   Raa &\to aaaR &
   bL &\to bRa &
   Rb &\to Lb &
   Rab &\to Lab.
  \end{align*}
In this system $bRa^nb$ rewrites to $bRa^{f(n)}b$ for $f$ defined by $f(2n) = 3n+1$
and $f(2n+1) = 3n+2$ for all $n$. This obviously yields an infinite reduction, but
$f$ is not linear, by which this example is outside the scope of ~\cite{oppe:2008}
and \cite{emme:enge:gies:2012}. In our approach a proof of non-termination  and even 
non-weak-normalization is extremely simple: $b a^* (L \mid R) a^* b$ is non-empty,
closed under rewriting and does not contain normal forms.
\end{example}

\subsection{Disproving Strong Normalization}

For disproving strong normalization based on Theorem~\ref{thm:sn}, the only difference with Technique~\ref{tec:aut:wn}
is that checking  that $L$ is closed under $\to$ by Lemma~\ref{lem:closed} has to be replaced by
checking  that $L$ is weakly closed under $\to$ by Lemma~\ref{lem:wclosed}.
The technique is applicable to string and term rewriting systems, 
and can be automated as described in Technique~\ref{tec:aut:wn} and Remark~\ref{sat:wclosed}.

\begin{example}\label{ex:lr:even}
  Let us consider the rewrite system
  \begin{align*}
   aaL &\to Laa & 
   Ra &\to aR &
   bL &\to bR &
   Rb &\to Lab &
   Rb &\to aLb
  \end{align*}
  This system is non-looping and non-terminating.
  However, in contrast to Example~\ref{ex:lr}, this system is weakly normalizing, since by always choosing the fourth rule the 
last rule is never used, and the first four rules are terminating.
  Hence the Technique~\ref{tec:aut:wn} is not applicable for this TRS.
However, the following pattern extends to an infinite reduction
\[ bR a^{2n} b \to^{2n} b a^{2n} Rb \to b a^{2n} Lab \to^{n} bL a^{2n+1} b \to \]
\[ bR a^{2n+1} b \to^{2n+1} b a^{2n+1} Rb \to b a^{2n+2} Lb \to^{n+1} bL a^{2n+2} b \to  bR a^{2n+2} b. \]

\noindent
\begin{minipage}{6cm}
Instead of finding this pattern explicitly, non-termination is also concluded from checking that $b (aa)^* (L \mid R \mid aR) a^* b$
describes a language satisfying all conditions from Theorem~\ref{thm:sn}. A corresponding 
automaton is given on the right.

\end{minipage}
\begin{minipage}{7cm}
\quad\includegraphics[scale=0.18]{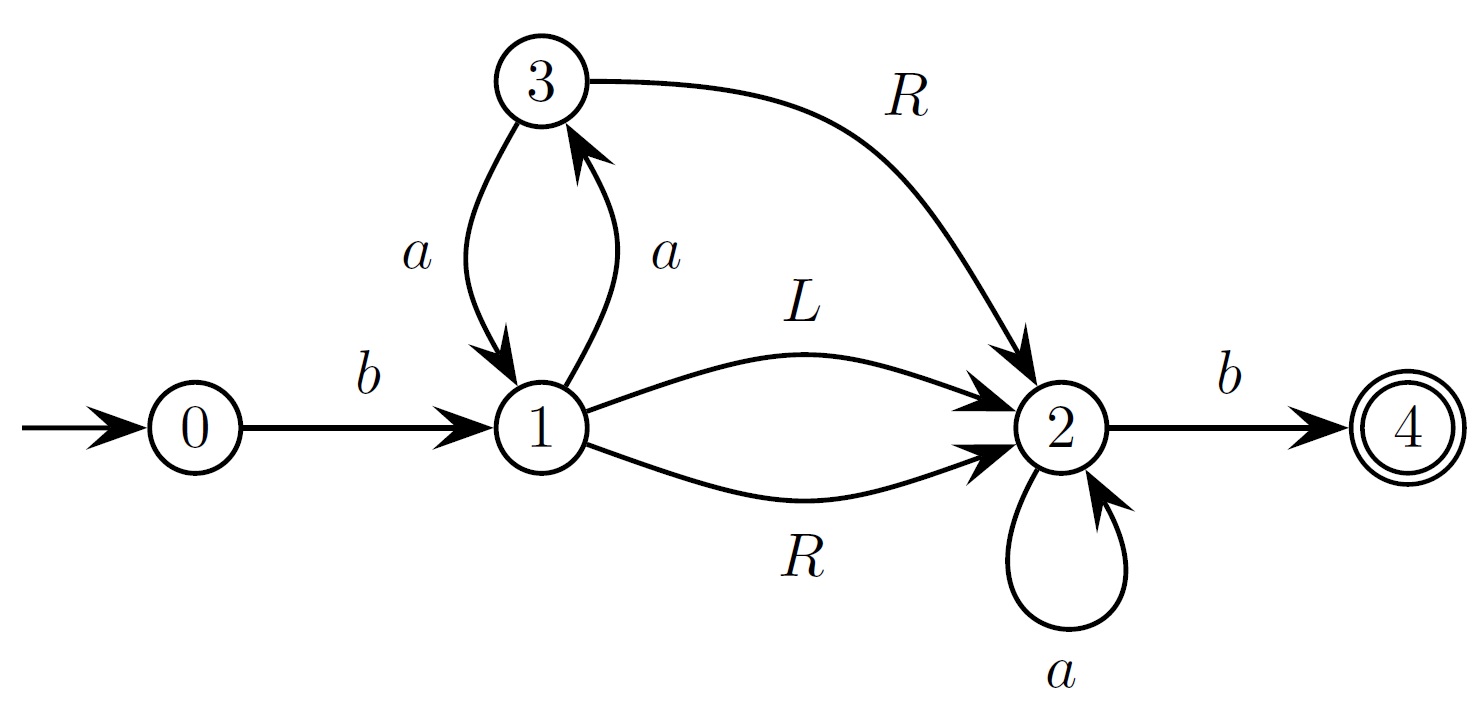}
\end{minipage}
\smallskip

\noindent  
The conditions of Theorem~\ref{thm:sn} are now 
checked as follows. Non-emptiness follows from
the existence of a path from state 0 to state 4. Every path from 0 to 4 either contains one of the
patterns $aaL$, $Ra$, $bL$ or $Rb$, so it remains to show weakly closedness under rewriting by
Lemma~\ref{lem:wclosed}. In the setting of string automata this means that for every left hand
side $\ell$ and every $\ell$-path from a state $p$ to a state $q$ we should find a $u$-path
from $p$ to $q$ for a string $u$ such that $\ell$ rewrites to $u$ in one or more steps.
For $\ell = aaL$ the only path is from 1 to 2, for which there is also an $Laa$ path.
For $\ell = Ra$ there is a path from 1 to 2, for which there is also an $aR$ path via 3.
The only other option for $\ell = Ra$ is a path from 3 to 2, for which there is also an $aR$ path 
via 1.
For $\ell = bL$ the only path is from 0 to 2, for which there is also a $bR$ path.
Finally, for $\ell = Rb$ there is a path from 1 to 4, for which there is also an $Lab$ path
and a path from 3 to 4, for which there is also an $aLb$ path, by which all conditions have
been verified. Note that for the $Rb$-path  from 1 to 4 it is essential to use the 4th rule,
while for the $Rb$-path  from 3 to 4 it is essential to use the last rule. 
\end{example}

This example can also be treated by the technique introduced in the following section.

%% file: automata-sat.tex
In this section, we collect decision procedures for the main properties of 
tree automata that we employ for proving non-termination,
and we describe how we encode these procedures as Boolean satisfiability problems (SAT).

\begin{encoding}[SAT encoding of tree automata]\label{sat:automaton}
  We encode the search for a tree automaton $A = \tuple{Q,\Sigma,F,\delta}$
  over a signature $\Sigma$ as a satisfiability problem as follows.
  We pick the number of states $n \in \nat$ the automaton should have; the set of states is $Q = \{s_1,\ldots,s_n\}$.
  While the set of states $Q$ is fix,
  we represent the final states $F \subseteq Q$ by $n$ fresh Boolean variables
  \begin{align*}
    v_{F,s_1},\; v_{F,s_2},\; v_{F,s_3},\ldots,v_{F,s_n}
  \end{align*}
  and, for every $f \in \Sigma$, we represent the transition relation $\delta$ by fresh variables
  \begin{align*}
    v_{f,q_1,\ldots,q_{\arity{f}},q} \quad\quad\text{for every $q_1,\ldots,q_{\arity{f}},q \in Q$}
  \end{align*}
  For the moment, there are no constraints (formulas) and 
  the interpretation of these variables can be chosen freely. 
  The intention is that $v_{F,s_i}$ is true if and only if $s_i$ is a final state,
  and $v_{f,q_1,\ldots,q_{\arity{f}},q}$ is true if and only if $f(q_1,\ldots,q_{\arity{f}}) \ared q$ is 
  a transition rule in $\delta$.
\end{encoding}

\begin{definition}
  A state $q \in Q$ of a tree automaton $A = \tuple{Q,\Sigma,F,\delta}$
  is called \emph{reachable} if there exists a ground term $t \in \ter{\Sigma}{\emptyset}$
  such that $t \ared^* q$.
\end{definition}

We assume, without loss of generality, that all states are reachable.
Note that requiring that all states are reachable is not a restriction
since we can always replace unreachable states by `copies' of reachable states.
We guarantee reachability as follows.

\begin{lemma}\label{lem:reachable}
  Let $A = \tuple{Q,\Sigma,F,\delta}$ be a tree automaton.
  Then all states of $A$ are reachable if and only if 
  there exists a total well-founded order $<$ on the states $Q$
  such that for every $q \in Q$ there exists $f \in \Sigma$
  and states $q_1 < q$, $q_2 < q$, \ldots, $q_{\arity{f}} < q$
  with $f(q_1,\ldots,q_{\arity{f}}) \ared q$.
\end{lemma}

\begin{proof}
  The reachable states are the smallest set $Q' \subseteq Q$
  that is closed under $\delta$, that is, 
  $q \in Q'$ whenever $f(q_1,\ldots,q_{\arity{f}}) \ared q$ for some $f \in \Sigma$ and states $q_1,\ldots,q_{\arity{f}} \in Q'$.

  For the `if'-part, assume that there was a non-reachable state.
  Let $q \in Q$ be the smallest non-reachable state with respect to the order $<$.
  By assumption there exist $f \in \Sigma$ and states $q_1 < q$, $q_2 < q$, \ldots, $q_{\arity{f}} < q$
  with $f(q_1,\ldots,q_{\arity{f}}) \ared q$.
  By choice of $q$ it follows that all states $q_1,q_2,\ldots,q_{\arity{f}}$ are reachable,
  and hence $q$ is reachable, contradicting the assumption.
  
  For the `only if'-part, assume that all states are reachable.
  Then $Q$ is the result of stepwise closing $\varnothing$ under $\delta$. 
  There exists a sequence of states $\varnothing = Q_0 \subseteq Q_1 \subseteq \ldots \subseteq Q_{|Q|} = Q$
  such that for every $0 \le i < |Q|$ we have $Q_{i+1} = Q_{i} \cup \{q_i\}$
  for some $q_i \in Q \setminus Q_i$ such that there are $f_i \in \Sigma$ and states $q_{i,1},\ldots,q_{i,\arity{f_i}} \in Q_i$
  with $f_i(q_{i,1},\ldots,q_{i,\arity{f}}) \ared q_i$.
  The order $<$ induced by $q_0 < q_1 < \ldots < q_{|Q|-1}$
  is a total order on the states with the desired property.
\end{proof}

\begin{encoding}[SAT encoding of reachability of all states]\label{sat:reachable}
  We extend the encoding of tree automata as described in Remark~\ref{sat:automaton}.
  We want to guarantee that all states are reachable by employing Lemma~\ref{lem:reachable}.
  However, instead of encoding an arbitrary well-founded relation,
  we make use of the fact that the names of states are irrelevant.
  Hence, without loss of generality (modulo renaming of states), we may assume that $s_1 < s_2 <\ldots < s_n$.
  We then encode the condition of Lemma~\ref{lem:reachable} by formulas
  \begin{align*}
    \bigvee_{f\in\Sigma,\;q_1 < q,\ldots,\;q_{\arity{n}}< q} v_{f,q_1,\ldots,q_n,q}
  \end{align*}
  for every $q \in Q$.
\end{encoding}

The following lemma is immediate.

\begin{lemma}\label{lem:empty}
  Let $A = \tuple{Q,\Sigma,F,\delta}$ be a tree automaton
  such that all states are reachable.
  Then $\lang{A} \ne \emptyset$ if and only if $F \ne \emptyset$.
\end{lemma}

\begin{encoding}[SAT encoding of $\lang{A} \ne \emptyset$]\label{sat:empty}
  In a setting where all states are reachable,
  the encoding of $\lang{A} \ne \emptyset$ as satisfiability problem
  trivializes to: $\bigvee_{q \in Q} v_{F,q}$.
\end{encoding}

The following lemma gives a simple criterion for closure under rewriting.
\newcommand{\cgen}{Genet~\cite[Proposition 12]{gene:1998}}
\begin{lemma}[\cgen]\label{lem:closed}
  Let $A = \tuple{Q,\Sigma,F,\delta}$ be a tree automaton and $R$ a left-linear term rewriting system.
  Then $\lang{A}$ is closed under rewriting with respect to $R$ if
  for every $\ell \to r \in R$, $\alpha : \vars \to Q$ and $q \in Q$ we have
  $\ell\alpha \ared_A^* q \;\implies\; r\alpha \ared_A^* q$.
\end{lemma}


Note that left-linearity of $R$ is crucial for the Lemma~\ref{lem:closed}
since $A$ can be a non-deterministic automaton.
If $R$ would contain non-left-linear rules $\ell \to r$
then we would need to check set-assignments $\alpha : \vars \to \powerset{Q}$
instead $\alpha : \vars \to Q$.
That is, we would need to take into account, 
that a non-deterministic automaton can interpret the same term by different states.

For terms $t$, we use $\varsof{t}$ to denote the set of variables occurring in $t$.
\begin{encoding}[SAT encoding of closure under rewriting]\label{sat:closed}
  We encode the conditions of~Lemma~\ref{lem:closed}. 
  Let the automaton $A$ be encoded as in Remark~\ref{sat:automaton}.
  Let $U$ be the set of all non-variable subterms 
  of left-hand sides and right-hand sides of rules in $R$.
  For every $t \in U$, assignment $\alpha : \varsof{t} \to Q$
  and $q \in Q$ we introduce a fresh variable
  \begin{align*}
    v_{t,\alpha,q} && \text{with the intended meaning: $v_{t,\alpha,q}$ is true $\iff$ $t\alpha \ared^* q$\;.}
  \end{align*}
  We ensure this meaning by the following formulas:
  for terms $t= f(t_1,\ldots,t_n) \in U$
  \begin{align*}
    v_{t,\alpha,q} \;\longleftrightarrow\; \bigvee_{q_1,\ldots,q_n \in Q} 
      \big( 
        v_{t_1,\alpha_1,q_1} \;\wedge\;
        \ldots
        \;\wedge\; v_{t_n,\alpha_n,q_n}
        \;\wedge\; v_{f,q_1,\ldots,q_n,q} 
      \big)
  \end{align*}
  where $\alpha_i$ is the restriction of $\alpha$ to the domain $\varsof{t_i}$.
  For variables $x \in U$, we stipulate $v_{x,\alpha,q} \iff \alpha(x) = q$;
  note that we can immediately evaluate and fill in these truth values.
  Finally, we encode $\ell\alpha \ared_A^* q \implies r\alpha \ared_A^* q$ by formulas
  \begin{align*}
    v_{\ell,\alpha,q} \;\to\; v_{r,\alpha,q}
  \end{align*}
  for every $\ell \to r \in R$, $\alpha : \varsof{\ell} \to Q$ and $q \in Q$.
\end{encoding}

The following modification of Lemma \ref{lem:closed} gives a simple criterion for weak
closure under rewriting.
The requirement $r\alpha \ared_A^* q$ of Lemma \ref{lem:closed}
is weakened to $t\alpha \ared_A^* q$ for some reduct $t$ the left-hand side $\ell$.

\begin{lemma}\label{lem:wclosed}
  Let $A = \tuple{Q,\Sigma,F,\delta}$ be a tree automaton and $R$ a left-linear term rewriting system.
  Then $\lang{A}$ is weakly closed under rewriting with respect to $R$ if 
  for every $\ell \to r \in R$, $\alpha : \vars \to Q$ and $q \in Q$ we have
  $\ell\alpha \ared_A^* q \;\implies\; t\alpha \ared_A^* q$ for some term $t$ such that $\ell \to_R^+ t$.
\end{lemma}
\begin{proof}
  Let $s \in \lang{A} \setminus \NF(\to_R)$. 
  Then $s = C[\ell\sigma]$ 
  for a context $C$, rewrite rule $\ell \to r \in R$ and substitution $\sigma : \vars \to \ter{\Sigma}{\emptyset}$.
  Since $s \in \lang{A}$ there exists $q \in Q$ such that $\ell\sigma \ared^* q$ and $C[q] \ared^* q'$ with $q'\in F$.
  By left-linearity $\ell$ does not contain duplicated occurrences of variables.
  As a consequence, there exists $\alpha : \vars \to Q$ such that $\sigma(x) \ared^* \alpha(x)$
  and $\ell\alpha \ared^* q$. 
  By the assumptions of the lemma, a term $t$ exists such that $\ell \to_R^+ t$ and   $t\alpha \ared^* q$.
  Hence $t\sigma \ared^* q$ and $C[t\sigma] \ared^* C[q] \ared^* q'$.
  Thus $C[t\sigma] \in \lang{A}$. Since $s = C[\ell\sigma] \to_R^+ C[t\sigma]$, this proves that $\lang{A}$ is 
  weakly closed under rewriting with respect to $R$.
\end{proof}

\begin{encoding}[SAT encoding of Lemma~\ref{lem:wclosed}]\label{sat:wclosed}
  The conditions of Lemma~\ref{lem:wclosed} can be encoded similar to Lemma~\ref{lem:closed} (described in Remark~\ref{sat:closed}). 
  In Lemma~\ref{lem:wclosed} the condition $r\alpha \ared_A^* q$
  is weakened to: $t\alpha \ared_A^* q$ for some reduct $t$ the left-hand side $\ell$.
  We can pick a finite set of reducts $U \subseteq \{t \mid \ell \to^+ t\}$ of the left-hand side $\ell$,
  and encode the disjunction $\bigvee_{t \in U} t\alpha \ared_A^* q$ as a Boolean satisfiability problem.
  Note that $U \ne \varnothing$ since $r \in U$.
\end{encoding}

Next, we want to guarantee that the language $\lang{A}$ contains no normal forms,
in other words, that every term in the language contains a redex.
For left-linear term rewriting systems $R$, we can reduce this problem
to language inclusion $\lang{A} \subseteq \lang{B}$
where $B$ is a tree automaton that accepts the language of reducible terms.
If $R$ is a left-linear rewrite system, then
the set of ground terms containing redex occurrences is a regular tree language.
A deterministic automaton $B$ for this language can be constructed
using the overlap-closure of subterms of left-hand sides, 
see further~\cite{endr:hend:2010,endr:hend:2011}.
Here, we do not repeat the construction, but state the lemma that we will employ:

\begin{lemma}\label{lem:redex:automaton}
  Let $\{\ell_1,\ldots,\ell_n\}$ be a set of linear terms over $\Sigma$. 
  Then we can construct a deterministic and complete automaton $B = \tuple{Q,\Sigma,F,\delta}$
  and sets $F_{\ell_1},\ldots,F_{\ell_n} \subseteq Q$ such that 
  for every term $t\in\ter{\Sigma}{\varnothing}$ and $i \in \{1,\ldots,n\}$ we have:
  \begin{itemize}
    \item $t \ared^* q$ with $q \in F_{\ell_i}$ if and only if $t$ is an instance of $\ell_i$.
  \end{itemize}
\end{lemma}
Note that by choosing $F = F_{\ell_1} \cup \ldots \cup F_{\ell_n}$ we obtain:
$t \ared^* q$ with $q \in F$ if and only if $t$ is an instance $\ell_i$ for some $i \in \{1,\ldots,n\}$.

\begin{example}\label{ex:redex:s}
  The following tree automaton $B_{S} = \tuple{Q,\Sigma,F,\delta}$ accepts the language of
  ground terms that contain a redex occurrence with respect to the $S$-rule $a(a(a(S, x), y), z) \to a(a(x, z), a(y, z))$.
  Here $Q = \{0,1,2,3\}$, $\Sigma = \{\sap,\combS\}$, $F = \{3\}$
  and
  \begin{align*}
    &\combS \ared 0 & 
    \sap(0,q) &\ared 1 &
    \sap(1,q) &\ared 2 & 
    \sap(2,q) &\ared 3 &
    \sap(3,q) &\ared 3 &
    \sap(q',3) &\ared 3 
  \end{align*}
  for all $q \in \{0,1,2\}$ and $q' \in \{0,1,2,3\}$.
  
  Since the automaton $B_{S}$ is deterministic and complete, 
  we can obtain an automaton $\overline{B_{S}} = \tuple{Q,\Sigma,\overline{F},\delta}$ 
  that accepts the complement of the language
  (the language of ground normal forms)
  by taking the complement $\overline{F} = \{0,1,2\}$ of the set of final states.
\end{example}

The following is crucial for feasibility of our approach.
Deciding language inclusion of non-deterministic automata is known to be {\small EXPTIME} complete, see~\cite{seid:1989}.
However, to guarantee that a language contains no normal forms,
it suffices to check whether two non-deterministic automata 
have a non-empty intersection.
This property can be decided in polynomial time by
constructing the product automaton and considering the reachable states.

\begin{definition}\label{def:product:automaton}
  The \emph{product $A \product B$} of tree automata $A = \tuple{Q,\Sigma,F,\delta}$ and $B = \tuple{Q',\Sigma,F',\delta'}$
  is the tree automaton $C = \tuple{Q \times Q',\Sigma,F\times F',\gamma}$ 
  where the transition relation $\gamma$ is given by
  \begin{align*}
    f(\;(q_1,p_1),\ldots,(q_n,p_n)\;) \ared_\gamma (q',p') 
    \;\;\iff\;\;
    f(q_1,\ldots,q_n) \ared_\delta q' 
    \;\wedge\;
    f(p_1,\ldots,p_n) \ared_{\delta'} p' 
  \end{align*}
  for every $f \in \Sigma$ of arity $n$ and states $q_1,\ldots,q_n,q' \in Q$ and $p_1,\ldots,p_n,p' \in Q'$.
\end{definition}

\begin{lemma}\label{lem:intersection}
  Let $A = \tuple{Q,\Sigma,F,\delta}$ and $B = \tuple{Q',\Sigma,F',\delta'}$ be tree automata.
  Then we have $\lang{A} \cap \lang{B} = \varnothing$
  if and only if
  in $A\product B$
  no state in $F \times F'$ is reachable.
\end{lemma}

\begin{proof}
  Let $A\times B = \tuple{Q \times Q',\Sigma,\emptyset,\gamma}$.
  For the `if'-part, assume that $\lang{A} \cap \lang{B} \ne \varnothing$. 
  Let $t \in \lang{A} \cap \lang{B}$.
  Then $t \ared_\delta^* q$ for some $q \in F$
  and $t \ared_{\delta'}^* q'$ for some $q' \in F'$.
  But then $t \ared_\gamma^* (q,q')$ 
  and hence $(q,q') \in F \times F'$ is reachable in $A\times B$; this contradicts the assumption. 

  For the `only if'-part, assume, for a contradiction, that $t \ared_{\gamma}^* (q,q')$ in $A \product B$ with $q \in F$ and $q' \in F'$. 
  Then this directly translates to $t \ared_\delta^* q$ in $A$ and $t \ared_{\delta'}^* q'$ in $B$.
  Hence $t \in \lang{A}$ and $t \in \lang{B}$,
  contradicting $\lang{A} \cap \lang{B} = \varnothing$.
\end{proof}

We can use Lemma~\ref{lem:intersection}
to check that the language $\lang{A}$ of an automaton $A$ does not contain normal forms.
To this end, we only need an automaton $B$ that accepts all ground normal forms.
Then $\lang{A}$ contains no normal forms if $\lang{A} \cap \lang{B} = \varnothing$.
\begin{example}
  The reachable states of the product $A_{S} \product \overline{B_{S}}$ 
  of the automata $A_S$ from Example~\ref{ex:automaton:s}
  and $\overline{B_{S}}$ from Example~\ref{ex:redex:s} are
  $(0,0), (1,1), (2,2), (2,1), (3,3), (3,2), (2,3), (4,3)$.
  The only state $(q,q')$ such that $q$ is accepting in $A_S$ is $(4,3)$ and
  $3$ is not an accepting state of $\overline{B_S}$.
  The conditions of Lemma~\ref{lem:intersection} are fulfilled and 
  hence $\lang{A_S} \cap \lang{\overline{B_S}} = \varnothing$.
  Recall that $\overline{B_{S}}$ accepts all ground normal forms, 
  and thus every term accepted by $A_S$ contains a redex.
\end{example}

\begin{encoding}[SAT encoding of empty intersection]\label{sat:intersection}
  Let $A = \tuple{Q,\Sigma,F,\delta}$ and $B = \tuple{Q',\Sigma,F',\delta'}$ be tree automata.
  Let $A\times B = \tuple{Q \times Q',\Sigma,\emptyset,\gamma}$.

  First, note that reachability of all states in the automata $A$ and $B$
  does not imply that all states in the product automaton $A\times B$ are reachable. 
  As a consequence, we have to `compute' the set of reachable states using Boolean satisfiability problems. 
  For this purpose, we reformulate Lemma~\ref{lem:intersection} in the following equivalent way:
  \ldots, then $\lang{A} \cap \lang{B} = \varnothing$
  if and only if
  there exists a set of states $P \subseteq Q \times Q'$
  such that
  \begin{enumerate}[label=(\roman*)]
    \item $P$ is \emph{closed under transitions} in $A \times B$, that is,
      $q \in P$ whenever $f(q_1,\ldots,q_n) \ared_\gamma q$
      for some $q_1,\ldots,q_n \in P$, and
    \item for all $(q,q') \in P$ it holds that $q \in F$ implies $q' \not\in F'$.
      \smallskip
  \end{enumerate}
  Note that this statement is equivalent to Lemma~\ref{lem:intersection}.
  Item (i) guarantees that $P$ contains all reachable states,
  and hence (ii) is required for at least the reachable states.
  Thus the conditions imply those of Lemma~\ref{lem:intersection}.
  On the other hand, we can take $P$ to be precisely the set of reachable states,
  and then the conditions are exactly those of Lemma~\ref{lem:intersection}.
  
  The idea is that the reformulated statement has a much more efficient
  encoding as Boolean satisfiability problem. 
  We only need to encode the closure of $P$ under transitions, 
  but there is no longer the need for encoding the property that $P$ is the smallest such set 
  (which is a statement of second-order logic).
  
  Assume that we have a SAT encoding of the automata $A$ and $B$ as in Remark~\ref{sat:automaton};
  we write $v_{A,\ldots}$ for the variables encoding $A$,
  and $v_{B,\ldots}$ for the variables encoding $B$.
  To represent the set $P$, we introduce variables
  $p_{(q,q')}$ for every $(q,q') \in Q \times Q'$ and the properties are translated into the following formulas:
  \begin{enumerate}[label=(\roman*)]
    \item for every $f\in\Sigma$ with arity $n$ and $(q_1,q'_1),\ldots,(q_n,q'_n),(q,q') \in Q\times Q'$:
      \begin{align*}
        (v_{A,f,q_1,\ldots,q_n,q} \;\wedge\; v_{B,f,q'_1,\ldots,q'_n,q'} 
         \;\wedge\; p_{(q_1,q_1')} \;\wedge\; p_{(q_2,q_2')} \;\wedge\; \ldots \;\wedge\; p_{(q_n,q_n')}) \;\to\; p_{(q,q')}
      \end{align*}
      
    \item for every $(q,q') \in Q\times Q'$: $(p_{(q,q')} \wedge v_{A,F,q}) \to \neg v_{B,F',q'}$.
      \smallskip
  \end{enumerate}
  Each of these formulas simplifies to a single clause (a disjunction of literals).
  
  We remark that we will employ this translation for the case that $B$ consists of the set of terms 
  containing redex occurrences with respect to a given rewrite system $R$.
  Then $B$ is known and fixed before the translation to a satisfiability problem.
  As a consequence, we know the truth values of $v_{B,f,q'_1,\ldots,q'_n,q'}$ and $v_{B,F',q'}$
  in the formulas above, and can immediate skip the generation
  of formulas that are trivially true (the large majority in case (i)).
\end{encoding} 

\begin{encoding}[Complexity of the SAT encoding]\label{rem:complexity}
  While the encoding is efficient for string rewriting systems,
  it suffers from an `encoding explosion' for term rewriting systems
  containing symbols of higher arity.
  The problem arises from the SAT encoding of the recursive computation of the interpretation of terms
  (described in Remark~\ref{sat:closed}).
  The computation of the interpretation of a term $f(t_1,\ldots,t_n)$
  containing $m$ variables needs $O(|Q|^{m+n+1})$ clauses:
  $m$ for the quantification over the variable assignments,
  $n$ for the possible states of $t_1,\ldots,t_n$
  and $1$ for the possible result states.
  To some extend, this problem can be overcome by `uncurrying' the system,
  that is, for every symbol $f$ of arity $n > 2$ we introduce 
  fresh symbols $f_1,\ldots,f_{n-1}$ of arity $2$ and
  then replace all occurrences of $f(t_1,\ldots,t_n)$ by $f_{n-1}(\ldots f_2(f_1(t_1,t_2),t_3)\ldots,t_n)$.
  This transformation helps to bring the complexity down to $O(|Q|^{m+3})$.
  Nevertheless, for example for the S-rule, which only contains binary symbols, we still need $|Q|^6$ clauses.
  We note that after the uncurrying transformation, an automaton with more states may be needed
  to generate `the same' language.
\end{encoding}

%% file: improved.tex
In this section, we improve the method for proving non-termination.
The methods introduced so far are not able to handle the following example.

\begin{example}\label{ex:lr:improved}
  We consider the following string rewriting system:
  \begin{align*}
   zL \to Lz &&
   Rz \to zR &&
   zLL \to zLR &&
   RRz \to LzRz
  \end{align*}
  This rewrite system is weakly normalizing but not strongly normalizing.
  The non-termination criteria introduced in the previous sections
  are not applicable for this system.
  Let us consider the first steps of an infinite reduction:
  \begin{align*}
    &\ \underline{zLL}zzRz\\
    \to&\ zLRzzRz \to zLzRzRz \to zLzzRRz\\
    \to&\ zLzzLzRz \to zLzLzzRz \to \underline{zLL}zzzRz\\
    \to&\ldots
  \end{align*}
  Note the underlined occurrences of $zLL$.
  Due to the rule $zLL \to zLR$,
  the word $zL$ is the marker for `turning' on the left;
  However, this marker $zL$ is itself a redex.
  To obtain an infinite reduction, this marker must not be reduced.
\end{example}

The idea for proving non-termination of systems like Example~\ref{ex:lr:improved}
is to let the automaton determine which redex to contract.
To this end, we introduce a `redex selection' function
\begin{align*}
  \rmap : Q \to \pow{R}
\end{align*}
that maps states of the automaton to sets of rules that may be contracted
at the corresponding position in the term.
The idea is that a redex $\ell\sigma$ in a term $C[\ell\sigma]$ with respect to a rule $\ell \to r$
is allowed to be contracted
if $\ell\sigma \ared^* q$ with $\ell\to r \in \rmap(q)$.
In this way, the automaton determines what redexes are to be contracted.
Then the automaton only needs to fulfill the property $\ell\alpha \ared_A^* q \implies r\alpha \ared_A^* q$
for the \emph{selected rules}:
\begin{align*}
  \ell \to r \in \rmap(q) \;\wedge\; \ell\alpha \ared_A^* q \quad\implies\quad 
  r\alpha \ared_A^* q
\end{align*}
for every rule $\ell \to r \in R$, state $q\in Q$ and $\alpha : \vars \to Q$.
Moreover, as proposed in~\cite{endr:vrij:wald:2009,endr:vrij:wald:2010,felg:thie:2014},
we weaken the requirement $r\alpha \ared_A^* q$ 
to $r\alpha \ared_A^* p$ for some $p \ge q$.
Here $\le$ is a quasi-order on the states and 
the automaton must be monotonic with respect to this order (see Definition~\ref{def:monotonic}).
The monotonicity guarantees that the language of the automaton is closed under rewriting.
For the present paper, this closure property holds only for the rules selected by $\rmap$.

\begin{definition}[Monotonicity]\label{def:monotonic}
  A tree automaton $A = \tuple{Q,\Sigma,F,\delta}$ 
  is \emph{monotonic} with respect to a quasi\nb-order $\le$ on the states $Q$
  if the following properties hold:
  \begin{enumerate}[label=(\emph{\roman*})]
    \item 
      For all $f \in \Sigma$ with arity $n$ and states $a_1 \le b_1$, $a_2 \le b_2$, \ldots, $a_n \le b_n$,  
      it holds
      \begin{align*}
        f(a_1,\ldots,a_n) \ared_A q \quad\implies\quad
        f(b_1,\ldots,b_n) \ared_A p \;\text{ for some $p\in Q$ with $q \le p$}
      \end{align*}
    \item 
      Whenever $q \in F$ and $q \le p$, then $p \in F$.
  \end{enumerate}
\end{definition}

The following lemma is immediate by induction on the size of the context.
\begin{lemma}\label{lem:monotonicity}
  Let $A = \tuple{Q,\Sigma,F,\delta}$ be a tree automaton
  that is monotonic with respect to a quasi\nb-order $\le$ on the states $Q$.
  Let $a,b \in Q$ with $a \le b$.
  Then for all ground contexts $C$ we have that $C[a] \ared a'$ with $a' \in Q$ implies that $C[b] \ared b'$ for some $b' \in Q$ with $a' \le b'$.
\end{lemma}

\begin{definition}[Runs]
  Let $A = \tuple{Q,\Sigma,F,\delta}$ be a tree automaton and $t \in \ter{\Sigma}{\emptyset}$.
  A \emph{run} of $A$ on $t$ is a function $\rho : \pos{t} \to Q$
  such that for every $p \in \pos{t}$ and $t(p) = f \in \Sigma$ 
  there is a rule $f(\rho(p1),\ldots,\rho(pn)) \ared \rho(p)$ in $\delta$.
  The run $\rho$ is accepting if $\rho(\varepsilon) \in F$.
\end{definition}

Note that there is a direct correspondence between runs on $t$ and rewrite sequences $t \ared^* q$.
We are now ready to state the generalized theorem for disproving strong normalization.
\begin{theorem}\label{thm:sn:improved}
  Let $R$ be a left-linear TRS. Let
  $A = \tuple{Q,\Sigma,F,\delta}$ be a tree automaton with $\lang{A} \ne \emptyset$,
  $\le$ a quasi-order on the states $Q$, and
  $\rmap : Q \to \pow{R}$ a function, called \emph{redex selection function}.
  Assume that the following properties hold:
  \begin{enumerate}[label=(\alph*)]
    \item \label{improved:monotonicity}
      The automaton $A$ is monotonic with respect to $\le$.
    \item \label{improved:model}
      For every state $q\in Q$, rule $\ell \to r \in R$ and $\alpha : \vars \to Q$ it holds that:
        \begin{align*}
          \ell\alpha \ared_A^* q \;\wedge\; \ell \to r \in \rmap(q) \quad\implies\quad 
          &(\exists p \in Q.\; q \le p \;\wedge\; r\alpha \ared_A^* p)
          \;\vee\;\\
          &(\exists r' \trianglelefteq r. \;\exists q' \in F.\; r'\!\alpha \ared_A^* q')
        \end{align*}
    \item \label{improved:reducible}
      For every term $t \in \ter{\Sigma}{\emptyset}$
      and accepting run $\rho$ on $t$ there is a position $p$
      such that $t|_p$ is an instance of the left-hand side of a rule $\ell \to r \in \rmap(\rho(p))$.
  \end{enumerate}
  Then $R$ is not strongly normalizing.
\end{theorem}

\newcommand{\treq}{\trianglerighteq}
\begin{proof}
  Assume that the conditions of the theorem are fulfilled. 
  To disprove strong normalization of $\to$
  it suffices to disprove strong normalization of $\to \circ \treq$
  where $\treq$ is the (non-strict) sub-term relation.
  We show that $\lang{A}$ and $\to \circ \treq$ fulfill the requirements of Theorem~\ref{thm:sn}.
  Let $t \in \lang{A}$.
  Then there exists an accepting run $\rho$ of $A$ on $t$.
  By item \ref{improved:reducible} there exists a position $p \in \pos{t}$
  and a rule $\ell \to r \in \rmap(\rho(p))$ such that $t|_p$ is an instance of $\ell$.
  Then $t|_p = \ell\sigma$ for some substitution $\sigma$.
  By left-linearity, we can define $\alpha : \varsof{\ell} \to Q$ by
  $\alpha(x) = \rho(pp')$ whenever $\ell|_{p'} \in \vars$.
  Then $\ell\alpha \ared^* \rho(p)$ and we distinguish cases according to item \ref{improved:model}:
  \begin{enumerate}
    \item 
      There exists $q \in Q$ with $\rho(p) \le q$ and $r\alpha \ared^* q$.
      We know that $t[\ell\sigma]_p = t \ared^* \rho(\varepsilon)$ and define $t' = t[r\sigma]_p$. 
      Note that $t \to t'$ and $t \to \circ \treq t'$.
      We have $t \ared t[\rho(p)] \ared \rho(\varepsilon)$ and $\rho(p) \le q$.
      By Lemma~\ref{lem:monotonicity} we have $t[q] \ared q'$ for some $q' \ge \rho(\varepsilon)$
      and by monotonicity~$q' \in F$.
      Hence $t' = t[r\sigma]_p \ared t[q] \ared^* q'$.
      Thus $t' \in \lang{A}$ and $t \to \circ \treq t'$.
    \item 
      There exist $r' \trianglelefteq r$ and $q \in F$ such that $r'\alpha \ared^* q$.
      Then $r'\sigma \ared^* q$ and hence $r'\sigma \in \lang{A}$.
      Moreover, $t \to\circ\treq r'\sigma$.
  \end{enumerate}
  This shows that $\lang{A}$ contains no normal forms and is weakly closed under $\to\circ\treq$.
  By Theorem~\ref{thm:sn}, $\to\circ\treq$ is not strongly normalizing and hence $\to$ is not strongly normalizing.
\end{proof}

\begin{encoding}[SAT encoding of the conditions of Theorem~\ref{thm:sn:improved}]\label{sat:improved}
  To encode the conditions of Theorem~\ref{thm:sn:improved} as a Boolean satisfiability problem
  we proceed as follows.
  For every $q,q' \in Q$, we introduce a fresh variable 
  \begin{align*}
    &v_{\le,q,q'} && \text{with the intended meaning: $v_{\le,q,q'}$ is true $\iff$ $q \le q'$\;,}
  \end{align*}
  and for every $q \in Q$ and $\ell \to r \in R$ a fresh variable
  \begin{align*}
    &v_{\rmap,q,\ell \to r} && \text{with the intended meaning: $v_{\rmap,q,\ell \to r}$ is true $\iff$ $\ell \to r \in \rmap(q)$\;.}
  \end{align*}
  The conditions~\ref{improved:monotonicity},~\ref{improved:model} and~\ref{improved:reducible}
  of Theorem~\ref{thm:sn:improved}
  can be encoded as follows:
  \begin{enumerate}[label=(\alph*)]
    \item 
      For condition~\ref{improved:monotonicity} we proceed as follows. 
      We ensure that $\le$ is a quasi-order using the following formulas
      \begin{align*}
        &v_{\le,q,q} && \text{for every $q \in Q$} && \text{(reflexivity)}\\
        &v_{\le,q_1,q_2} \wedge v_{\le,q_2,q_3} \to v_{\le,q_1,q_3} && \text{for every $q_1,q_2,q_3 \in Q$} && \text{(transitivity)}
      \end{align*}
      For monotonicity of the automaton with respect to $\le$ we introduce fresh variables
      \begin{align*}
        m_{f,q_1,\ldots,q_{\arity{f}},q,q'}
      \end{align*}
      and formulas
      \begin{align*}
        m_{f,q_1,\ldots,q_{\arity{f}},q,q'} \;\to\; v_{\le,q,q'} \wedge v_{f,q_1,\ldots,q_{\arity{f}},q'}
      \end{align*}
      for every $f \in \Sigma$ and $q_1,\ldots,q_{\arity{f}},q,q' \in Q$.
      Then monotonicity translates to formulas
      \begin{align*}
        &v_{F,q} \wedge v_{\le,q,q'} \to v_{F,q'}
      \end{align*}
      for every $q,q' \in Q$, and
      \begin{align*}
        &v_{f,q_1,\ldots,q_{\arity{f}},q} \wedge v_{\le,q_i,q'_i} \;\to\; \bigvee_{q' \in Q} m_{f,q_1,\ldots,q_{i-1},q_i',q_{i+1},\ldots,q_{\arity{f}},q,q'}
      \end{align*}
      for every $q_1,\ldots,q_{\arity{f}},q \in Q$, $i \in \{1,\ldots,\arity{f}\}$ and $q_i' \in Q$.
    \medskip
    
    \item 
      For condition~\ref{improved:model} we adapt Remark~\ref{sat:closed}.
      For every $t \in U$, assignment $\alpha : \varsof{t} \to Q$ and $q \in Q$,
      we have fresh variables
      \begin{align*}
        v_{t,\alpha,q} && \text{with the intended meaning: $v_{t,\alpha,q}$ is true $\iff$ $t\alpha \ared^* q$\;.}
      \end{align*}
      with the corresponding formulas (ensuring the intended meaning) as in Remark~\ref{sat:closed}.
      For every $\ell \to r \in R$, $\alpha : \varsof{\ell} \to Q$ and $q,q' \in Q$,
      we introduce a fresh variable $m_{r,\alpha,q,q'}$ and formula
      \begin{align*}
        &m_{r,\alpha,q,q'} \;\to\; v_{\le,q,q'} \wedge v_{r,\alpha,q'} \;.
      \end{align*}
      Likewise, for every $\ell \to r \in R$, $r' \trianglelefteq r$, $\alpha : \varsof{r'} \to Q$ and $q \in Q$,
      we add a fresh variable $m_{F,r',\alpha,q}$ and formula
      \begin{align*}
        &m_{F,r',\alpha,q} \;\to\; v_{F,q} \wedge v_{r,\alpha,q}
      \end{align*}
      Then condition~\ref{improved:model} translates to
      \begin{align*}
        &v_{\ell,\alpha,q} \wedge v_{\rmap,q,\ell \to r}
            \;\to\; \big( \bigvee_{q' \in Q} m_{r,\alpha,q,q'} \big) \vee 
            \big( \bigvee_{r' \trianglelefteq r,\; q\in Q} m_{F,r',\alpha,q} \big)
      \end{align*}
      for every $\ell \to r \in R$, $\alpha : \varsof{\ell} \to Q$ and $q \in Q$.

    \medskip
    
    \item
      For condition~\ref{improved:reducible} we extend Remark~\ref{sat:intersection}, as follows.
      Let $\ell_1,\ldots,\ell_n$ be the left-hand sides of rules in $R$.
      Let $B$ be the automaton and $F_{\ell_1},\ldots,F_{\ell_n}$ the sets of states obtained from Lemma~\ref{lem:redex:automaton}.
      We construct the product automaton $A \product B$,
      and then we compute those states that are reachable without passing states $(q,q')$ for which there exists $\ell \to r \in \rmap(q)$
      such that $q' \in F_\ell$ (that is, the rule $\ell \to r$ is selected by $A$ and $B$ confirms that the term is an instance of $\ell$).
      We adapt the encoding of Remark~\ref{sat:intersection} as follows:
      \begin{enumerate}[label=(\roman*)]
        \item for every $f\in\Sigma$ with arity $n$ and $(q_1,q'_1),\ldots,(q_n,q'_n),(q,q') \in Q\times Q'$:
          \begin{align*}
            \big(\;\;v_{A,f,q_1,\ldots,q_n,q} &\;\wedge\; v_{B,f,q'_1,\ldots,q'_n,q'} 
             \;\wedge\; p_{(q_1,q_1')} \;\wedge\; p_{(q_2,q_2')} \;\wedge\; \ldots \;\wedge\; p_{(q_n,q_n')} \\
             &\;\wedge\; \framebox{$\displaystyle \bigwedge_{\ell\to r \in R,\; q' \in F_{\ell}} \neg v_{\rmap,q,\ell \to r}$}
             \;\;\big) \;\to\; p_{(q,q')}
          \end{align*}
      
        \item for every $(q,q') \in Q\times Q'$: $p_{(q,q')} \to \neg v_{A,F,q}$.
          \smallskip
      \end{enumerate}
      In (i), due to the added condition, we do not consider all reachable states $(q,q')$ 
      but only those for which there exists no rule $\ell \to r \in R$ 
      that is applicable at the root of the current subterm ($q' \in F_\ell$)
      and that is activated ($v_{\rmap,q,\ell \to r}$).
      The formulas in (ii) guarantee that no state $(q,q')$, that is reachable in this sense,
      is accepted by $A$.
  \end{enumerate}
\end{encoding}

\begin{encoding}[Incompleteness with respect to Loops]\label{rem:loops}
  We note there exist (left-linear) TRSs systems that admit loops but
  cannot be proven non-terminating using Theorem~\ref{thm:sn:improved}.
  For example, consider the following rewrite system:
  \begin{align*}
    a &\to b & f(a,b) &\to g(a) & g(x) \to f(x,x)
  \end{align*}
  This system admits the loop $g(a) \to f(a,a) \to f(a,b) \to g(a)$.
  However, it is not possible to prove non-termination of this system by Theorem~\ref{thm:sn:improved}.
  We sketch the argument.
  Assume that there was a tree automaton $A$ fulfilling the conditions of the theorem.
  Then the automaton admits the rewrite step $g(a) \to f(a,a)$.
  Since both $a$'s in $f(a,a)$ are copies of the $a$ in $g(a)$,
  the automaton `interprets' both occurrences of $a$ as the same state.
  Now, to obtain an infinite rewrite sequence, the rule $a \to b$ must be activated for the right $a$ in $f(a,a)$.
  However, then this rule is also activated for the left $a$,
  but contracting the left $a$ leads to a normal form.
  
\end{encoding}

%% file: results.tex
We have implemented the improved method for disproving strong normalization (Theorem~\ref{thm:sn:improved})
presented in this paper.
For the purpose of evaluating our techniques,
the tool applies only the methods presented in this paper,
and no other non-termination method like loop checks.
The SAT solver employed for the evaluation results in this section is MiniSat~\cite{minisat:05}.
Our tool can be downloaded from~\url{http://joerg.endrullis.de/non-termination/}.

Our tool can automatically prove non-termination of all examples in this paper, including the S-rule and the $\delta$-rule.
The following table shows the size of the automata that are found by the tool as witnesses for non-termination for the examples in our paper:
\begin{center}
  \begin{tabular}{|c|c|c|c|c|c|c|c|}
    \hline
    Example & \ref{ex:lr} & \ref{ex:S} & \ref{ex:delta} & \ref{ex:lr:aaa} & \ref{ex:lr:even} & \ref{ex:lr:improved} \\
    \hline
    Number of states & 4 & 5 & 3 & 4 & 5 & 6\\
    \hline
  \end{tabular}
\end{center}
Each of these automata has been found within less than a second on a dual core laptop.

We have also evaluated our methods on the database used in~\cite{emme:enge:gies:2012},
consisting of 58 non-terminating term rewriting systems that do not admit loops.
The tool AProVE recognizes 44 systems as non-terminating; an impressive 76\%.
An extension of AProVE with our method would increase the recognition by 8.5\% to 84.5\%.
In other words, our method succeeds on 36\% (that is 5 systems) 
of the remaining 14 systems for which AProVE did not find a proof.
These 5 systems are:
\begin{itemize}
  \item \texttt{nonloop/TRS/emmes/ex3\_4.trs}
  \item \texttt{nonloop/TRS/own/challenge\_fab.trs}
  \item \texttt{nonloop/TRS/own/downfrom.trs}
  \item \texttt{nonloop/TRS/own/ex\_payet.trs}
  \item \texttt{nonloop/TRS/own/isList-List.trs}
\end{itemize}
In total, our tool succeeds for 26 of the 58 non-looping examples from~\cite{emme:enge:gies:2012}.
The results suggest that our method and that of~\cite{emme:enge:gies:2012}
are complementary and should be combined for maximum strength.
The paper~\cite{emme:enge:gies:2012} explicitly mentions that the following
example is beyond their techniques (this example is not part of the database above):
\newcommand{\true}{\mathrm{true}}
\newcommand{\successor}{\mathrm{s}}
\newcommand{\isNat}{\mathit{isNat}}
\newcommand{\double}{\mathit{double}}
\begin{align*}
  f(\true,\true,x,\successor(y)) &\to f(\isNat(x),\isNat(y),\successor(x),\double(\successor(y))) \\
  \isNat(0) &\to \true \\
  \isNat(\successor(x)) &\to \isNat(x) \\
  \double(0) &\to 0 \\
  \double(\successor(x)) &\to \successor(\successor(\double(x)))
\end{align*}
Our non-termination techniques can handle this system: 
the tool finds an automaton with $6$ states within 3 seconds
(using the transformation from Remark~\ref{rem:complexity}).

Finally, we have evaluated the tool on the termination problem database (TPDB).
We have run our tool on all string and term rewriting systems (of the standard categories) 
that remained unsolved during the last full run of all tools in December 2013.
For string rewriting, our tool was able to disprove termination for 13,
and for term rewriting, for 8 systems of the unsolved systems.
This corresponds to an increase of strength of 11.5\% (114 + 13) for string rewriting and 
of 3\% (274 + 8) for term rewriting.
Let us mention that many of the 13 string rewriting systems
actually admit loops, but very complicated ones, that are not found by the standard tools.
These loops have been found in previous competitions by the tools Matchbox~\cite{wald:2004}
and Knocked for Loops~\cite{kfl:2010}.

%% file: conclusions.tex
In this paper, we have employed regular languages for proving non-termination.
Instead of searching for an infinite reduction explicitly we search for a regular language
with properties from which non-termination easily follows. After encoding these properties in a propositional formula,
the actual search is done by a SAT solver. In some examples, like
Example \ref{ex:lr:aaa}, a very simple corresponding regular language is quickly found by our 
approach, while the actual infinite reductions have a non-linear pattern being beyond earlier approaches.

For future work, it is interesting to investigate whether this approach can be
extended to context-free (tree) languages;
such an approach could potentially also generalize~\cite{emme:enge:gies:2012}.
The question is whether there are efficient criteria to check the conditions of Theorem~\ref{thm:sn}. 
For example, consider the following string rewriting system:
\begin{align*}
  && bB &\to Bb &
  bcd &\to BcD & 
  Dd &\to dD \\
  aX &\to abb &
  BX &\to Xb &
  bcd &\to XcY & 
  YD &\to dY &
  Ye &\to dde &
\end{align*}
This system admits for every $n > 1$ reductions of the form
\begin{align*}
  &a\,b^n\,c\,d^n\,e 
  \to^* a\,B^{n-1}\,bcd\,D^{n-1}\,e 
  \to^* a\,B^{n-1}\,XcY\,D^{n-1}\,e 
  \to^* a\,b^{n+1}\,c\,d^{n+1}\,e
\end{align*}
As a description of this pattern needs a context-free language, it is unlikely that a regular language exists 
that fulfills the requirements of Theorem~\ref{thm:sn}. 

As described in Remark~\ref{rem:complexity},
the SAT encoding of (non-deterministic) automata is not efficient for symbols of higher arity.
We think that these problems can be overcome by more efficient encodings of automata.
For example, the uncurrying transformation mentioned in Remark~\ref{rem:complexity} can be seen 
as a restriction of the shape of the automata (the transition is computed argument by argument)
instead of a transformation on the system.
It would be interesting to investigate what other restrictions 
would lead to a more efficient representation of automata as Boolean satisfiability problems.
Results in this direction can be of interest in various areas where
automata are applied.

We think that it is also interesting to investigate whether
the characterization of strong and weak normalization (Theorems~\ref{thm:sn} and~\ref{thm:wn})
can be adapted to the setting of  
infinitary rewriting~\cite{endr:hend:klop:2012,endr:polo:2012b,endr:hans:hend:polo:silv:2015}
with infinite terms and ordinal-length reductions;
the interesting properties then are infinitary strong and weak normalization.

Finally, we note that equality of streams~\cite{endr:hend:bakh:2012,endr:hend:bakh:rosu:2014,zant:endr:2011,endr:hend:bodi:2013} 
(infinite sequences of symbols) can be rendered as a non-termination problem
(a comparison program running indefinitely if the streams are equal, and terminating as soon as a difference is found).
It remains to be investigated whether non-termination techniques can 
be employed fruitfully for proving stream equality.

%% file: main.bbl
\begin{thebibliography}{10}

\bibitem{aprove}
Homepage of {AProVE}, 2015.
\newblock \url{http://aprove.informatik.rwth-aachen.de}.

\bibitem{ttt2}
Homepage of {TTT2}, 2015.
\newblock \url{http://cl-informatik.uibk.ac.at/software/ttt2/}.

\bibitem{starling:2015}
H.P. Barendregt, J.~Endrullis, J.W. Klop, and J.~Waldmann.
\newblock {Dance of the Starlings}.
\newblock 2015.
\newblock To be published on arXiv.org.

\bibitem{starling:owl:2015}
H.P. Barendregt, J.~Endrullis, J.W. Klop, and J.~Waldmann.
\newblock {Songs of the Starling and the Owl}.
\newblock 2015.
\newblock To be published on arXiv.org.

\bibitem{cook}
B.~Cook.
\newblock Priciples of program termination.
\newblock
  \url{http://research.microsoft.com/en-us/um/cambridge/projects/terminator/principles.pdf}.

\bibitem{minisat:05}
N.~E{\'e}n and A.~Biere.
\newblock {Effective Preprocessing in SAT through Variable and Clause
  Elimination}.
\newblock In {\em Proc.\ Conf.\ on Theory and Applications of Satisfiability
  Testing (SAT '05)}, volume 3569 of {\em LNCS}, pages 61--75. Springer, 2005.

\bibitem{emme:enge:gies:2012}
F.~Emmes, T.~Enger, and J.~Giesl.
\newblock {Proving Non-looping Non-termination Automatically}.
\newblock In {\em International Joint Conference on Automated Reasoning (IJCAR
  2012)}, volume 7364 of {\em LNCS}, pages 225--240. Springer, 2012.

\bibitem{endr:vrij:wald:2009}
J.~Endrullis, R.~C. de~Vrijer, and J.~Waldmann.
\newblock {Local Termination}.
\newblock In {\em Proc.\ Conf.\ on Rewriting Techniques and Applications
  (RTA)}, volume 5595 of {\em LNCS}, pages 270--284. Springer, 2009.

\bibitem{endr:vrij:wald:2010}
J.~Endrullis, R.C. de~Vrijer, and J.~Waldmann.
\newblock {Local Termination: Theory and Practice}.
\newblock {\em Logical Methods in Computer Science}, 6(3), 2010.

\bibitem{endr:geuv:simo:zant:2011}
J.~Endrullis, H.~Geuvers, J.~G. Simonsen, and H.~Zantema.
\newblock {Levels of Undecidability in Rewriting}.
\newblock {\em {Information and Computation}}, 209(2):227--245, 2011.

\bibitem{endr:geuv:zant:2009}
J.~Endrullis, H.~Geuvers, and H.~Zantema.
\newblock {Degrees of Undecidability in Term Rewriting}.
\newblock In {\em {Proc.\ Int.\ Workshop on Computer Science Logic
  (CSL~2009)}}, volume 5771 of {\em LNCS}, pages 255--270. Springer, 2009.

\bibitem{endr:grab:hend:klop:vrij:2009}
J.~Endrullis, C.~Grabmayer, D.~Hendriks, J.W. Klop, and R.C. de~Vrijer.
\newblock {Proving Infinitary Normalization}.
\newblock In {\em Postproc.\ Int.\ Workshop on Types for Proofs and Programs
  (TYPES 2008)}, volume 5497 of {\em LNCS}, pages 64--82. Springer, 2009.

\bibitem{endr:grab:klop:oost:2011}
J.~Endrullis, C.~Grabmayer, J.W. Klop, and V.~van Oostrom.
\newblock {On Equal $\mu$-Terms}.
\newblock {\em Theoretical Computer Science}, 412(28):3175--3202, 2011.

\bibitem{endr:hans:hend:polo:silv:2015}
J.~Endrullis, H.~Hvid Hansen, D.~Hendriks, A.~Polonsky, and A.~Silva.
\newblock {A Coinductive Framework for Infinitary Rewriting and Equational
  Reasoning}.
\newblock In {\em Proc.\ Conf.\ on Rewriting Techniques and Applications (RTA
  2015)}, Leibniz International Proceedings in Informatics. Schloss Dagstuhl,
  2015.

\bibitem{endr:hend:2010}
J.~Endrullis and D.~Hendriks.
\newblock {Transforming Outermost into Context-Sensitive Rewriting}.
\newblock {\em Logical Methods in Computer Science}, 6(2), 2010.

\bibitem{endr:hend:2011}
J.~Endrullis and D.~Hendriks.
\newblock {Lazy Productivity via Termination}.
\newblock {\em Theoretical Computer Science}, 412(28):3203--3225, 2011.

\bibitem{endr:hend:bakh:2012}
J.~Endrullis, D.~Hendriks, and R.~Bakhshi.
\newblock {On the Complexity of Equivalence of Specifications of Infinite
  Objects}.
\newblock In {\em ACM SIGPLAN International Conference on Functional
  Programming (ICFP 2012)}, pages 153--164. ACM, 2012.

\bibitem{endr:hend:bakh:rosu:2014}
J.~Endrullis, D.~Hendriks, R.~Bakhshi, and G.~Ro\c{s}u.
\newblock On the complexity of stream equality.
\newblock {\em Journal of Functional Programming}, 24(2--3):166--217, 2014.

\bibitem{endr:hend:bodi:2013}
J.~Endrullis, D.~Hendriks, and M.~Bodin.
\newblock {Circular Coinduction in Coq Using Bisimulation-Up-To Techniques}.
\newblock In {\em Proc. Conf. on Interactive Theorem Proving (ITP)}, volume
  7998 of {\em LNCS}, pages 354--369. Springer, 2013.

\bibitem{endr:hend:klop:2012}
J.~Endrullis, D.~Hendriks, and J.W. Klop.
\newblock {Highlights in Infinitary Rewriting and Lambda Calculus}.
\newblock {\em Theoretical Computer Science}, 464:48--71, 2012.

\bibitem{endr:hofb:wald:2006}
J.~Endrullis, D.~Hofbauer, and J.~Waldmann.
\newblock {Decomposing Terminating Rewrite Relations}.
\newblock In {\em Proc. Workshop on Termination (WST '06)}, pages 39--43, 2006.

\bibitem{endr:polo:2012b}
J.~Endrullis and A.~Polonsky.
\newblock {Infinitary Rewriting Coinductively}.
\newblock In {\em Proc.\ Types for Proofs and Programs (TYPES 2012)}, volume~19
  of {\em Leibniz International Proceedings in Informatics}, pages 16--27.
  Schloss Dagstuhl, 2013.

\bibitem{endr:zant:2015}
J.~Endrullis and H.~Zantema.
\newblock {Proving Non-Termination by Finite Automata}.
\newblock In {\em Proc.\ Conf.\ on Rewriting Techniques and Applications (RTA
  2015)}, Leibniz International Proceedings in Informatics. Schloss Dagstuhl,
  2015.

\bibitem{felg:thie:2014}
B.~Felgenhauer and R.~Thiemann.
\newblock {Reachability Analysis with State-Compatible Automata}.
\newblock In {\em LATA}, volume 8370 of {\em LNCS}, pages 347--359. Springer,
  2014.

\bibitem{gene:1998}
T.~Genet.
\newblock {Decidable Approximations of Sets of Descendants and Sets of Normal
  Forms}.
\newblock In {\em Proc. Conf. on Rewriting Techniques and Applications (RTA
  '98)}, volume 1379 of {\em LNCS}, pages 151--165. Springer, 1998.

\bibitem{gese:hofb:wald:zant:2007}
A.~Geser, D.~Hofbauer, J.~Waldmann, and H.~Zantema.
\newblock On tree automata that certify termination of left-linear term
  rewriting systems.
\newblock {\em Information and Computation}, 205(4):512--534, 2007.

\bibitem{gese:zant:1999}
A.~Geser and H.~Zantema.
\newblock {Non-looping String Rewriting}.
\newblock {\em {RAIRO Theoretical Informatics and Applications}},
  33(3):279--302, 1999.

\bibitem{korp:midd:2009}
M.~Korp and A.~Middeldorp.
\newblock Match-bounds revisited.
\newblock {\em Information and Computation}, 207(11):1259--1283, 2009.

\bibitem{mous:lada:zant:2010}
M.~Mousazadeh, B.~T. Ladani, and H.~Zantema.
\newblock Liveness verification in trss using tree automata and termination
  analysis.
\newblock {\em Computing and Informatics}, 29(3):407--426, 2010.

\bibitem{oppe:2008}
M.~Oppelt.
\newblock {Automatische Erkennung von Ableitungsmustern in nichtterminierenden
  Wortersetzungssystemen}.
\newblock Technical report, HTWK Leipzig, Germany, 2008.
\newblock Diploma Thesis.

\bibitem{seid:1989}
Helmut Seidl.
\newblock Deciding equivalence of finite tree automata.
\newblock In {\em STACS 89}, volume 349 of {\em LNCS}, pages 480--492.
  Springer, 1989.

\bibitem{wald:2000}
J.~Waldmann.
\newblock {The Combinator S}.
\newblock {\em Information and Computation}, 159(1--2):2--21, 2000.

\bibitem{wald:2004}
J.~Waldmann.
\newblock {Matchbox: A Tool for Match-Bounded String Rewriting}.
\newblock In {\em Proc. Conf. on Rewriting Techniques and Applications (RTA
  '04)}, volume 3091 of {\em LNCS}, pages 85--94. Springer, 2004.

\bibitem{wald:2012}
J.~Waldmann.
\newblock Compressed loops (draft), 2012.

\bibitem{zank:midd:2007}
H.~Zankl and A.~Middeldorp.
\newblock {Nontermination of String Rewriting using SAT}, 2007.

\bibitem{zank:ster:hofb:midd:2010}
H.~Zankl, C.~Sternagel, D.~Hofbauer, and A.~Middeldorp.
\newblock Finding and certifying loops.
\newblock In {\em Proc. Conf. on Theory and Practice of Computer Science
  (SOFSEM 2010)}, volume 5901 of {\em LNCS}, pages 755--766. Springer, 2010.

\bibitem{kfl:2010}
H.~Zankl, C.~Sternagel, D.~Hofbauer, and A.~Middeldorp.
\newblock Finding and certifying loops.
\newblock In {\em SOFSEM 2010: Theory and Practice of Computer Science}, volume
  5901 of {\em LNCS}, pages 755--766. Springer, 2010.

\bibitem{zant:endr:2011}
H.~Zantema and J.~Endrullis.
\newblock {Proving Equality of Streams Automatically}.
\newblock In {\em Proc.\ Conf.\ on Rewriting Techniques and Applications (RTA
  2011)}, pages 393--408, 2011.

\end{thebibliography}
